\documentclass[submission,copyright,creativecommons]{eptcs}

\usepackage{graphicx,proof} 
\usepackage{amsthm} 
\usepackage{amsmath}
\usepackage{amssymb}
\usepackage{iftex}

\ifpdf
  \usepackage{underscore}         
  \usepackage[T1]{fontenc}        
\else
  \usepackage{breakurl}           
\fi

\newtheorem{theorem}{Theorem}[section]

\newtheorem{lemma}[theorem]{Lemma}

\newtheorem{definition}[theorem]{Definition}

\newtheorem{example}[theorem]{Example}

\newcommand{\Lens}{\msf{Lens}}

\newcommand{\tensor}{\otimes}

\newcommand{\Sel}{\msf{Sel}}

\usepackage{xspace}

\newcommand{\abs}[1]{|#1|}

\newcommand{\msf}[1]{\mbox{{\sf #1}}}

\newcommand{\T}{\msf{T}}
\newcommand{\R}{\msf{R}}
\renewcommand{\P}{\msf{P}}
\newcommand{\AM}{\msf{argmax}}

\title{A Category Theoretic Approach to \\
Compositional Approximation in Game Theory}
\author{Neil Ghani
\institute{MSP Group, University of Strathclyde, Scotland}
\email{ng@cis.strath.ac.uk}
}
\def\titlerunning{Categories Theory for Approximate  Game Theory}
\def\authorrunning{Neil Ghani}
\begin{document}
\maketitle

\begin{abstract}

\noindent 
This paper uses category theory to guide the development of an entirely new approach to approximate game theory. Game theory is the study of how different agents within a multi-agent system take decisions. At its core, game theory asks what an optimal decision is in a given scenario. Thus approximate game theory asks what is an approximately optimal decision in a given scenario. This is important in practice as --- just like in much of computing --- exact answers maybe too difficult (or even impossible) to compute given inherent uncertainty in input. \\[1ex]

\noindent We consider first {\em Selection Functions} which are a simple model of compositional game theory.  We develop i) a simple yet robust model of approximate equilibria; ii) the algebraic properties of approximation with respect to selection functions; and iii)  relate approximation to the compositional structure of selection functions. We then repeat this process for {\em Open Games} --- a more advanced model of game theory featuring the key operation of sequential composition of games. Finally, we use approximation to develop new metrics on game theory and show these yield core theorems in what one might term  "Metric Game Theory".
\end{abstract}

\section{Introduction}
This paper uses category theory to develop 
an entirely new approach to compositional approximate game theory. Approximation in computation is ubiquitous e.g.\ it arises in: i) stochastic systems, where one merely has probability distributions over values induced by inherent/simulated randomness; ii)  resource limited environments, where exact computation is prohibitively expensive; iii) systems with imperfect/partial recall, where one only has limited information about what has happened or the intentions/trustworthiness of each agent; and iv) non-exact computation where primitive data (e.g.\ from sensors) is inexact and supplied with error bars. These scenarios arise in e.g. cyber-physical systems, machine learning, robotics, automotive engineering, aerospace, and energy systems. \\[1ex]

\noindent On the other hand, game theory is the study of how different agents within a multi-agent system take decisions. The simplest model of game theory we consider are the {\em multi-valued selection functions}~\cite{escardo-oliva} which are functions  $s:(X \to V) \to PX$. Henceforth we merely call them selection functions. Here, i) the set $X$ is thought of as a set of actions (or moves, or choices) the agent can consider; ii) the set $V$ is thought of as a type of utilities; and iii) a function $X \to V$ is called a {\em utility function} as it assigns, to each possible action, a utility measuring how good that action is deemed to be. A selection function then maps 
{\em every} utility function to those actions deemed optimal for that utility function as as such is a model of game theory. The canonical selection function is $\AM$ defined by
\[
x \in \AM \; f \;\; \mbox{ iff } \;\; 
(\forall x':X) \; f x \geq f x' 
\]
Another selection function is $\epsilon\!\!-\!\!\AM$ defined by 
\[
x \in \epsilon\!\!-\!\!\AM \; f \;\; \mbox{ iff } \;\; 
(\forall x':X) \; f x \geq f x' - \epsilon
\]
which means for an action $x$ to be $\epsilon$-optimal it need merely have an associated utility which is within $\epsilon$ of a true optimal action. This approach aligns with the general approach to approximation of uncertainty budgets as we track - and therefore can check - how much utility is lost by moving to an approximately optimal decision. Notice this is NOT saying $x$ is close to an optimal choice of action, rather that the utility of $x$ is close to being the utility of an optimal action. One might study the former and, indeed, this was our first intention. However, the mathematical behaviour of that concept was poorer so we switched to the approach taken in this paper.  Nevertheless, this was the first of several moments where key decisions need to be taken to obtain a well behaved mathematical theory showing the combination of approximation and game theory to be a delicate and nuanced endeavour.  \\[1ex]

\noindent A second difficulty arises 
from the desire for {\em compositionality} to ensure scalability. Compositional systems are those where larger systems can be defined by assembling smaller subsystems together which means that proving properties of larger systems can be reduced to proving properties of smaller (and hence easier to reason about) subsystems. At its core, compositionality is  a story about structure and so the natural mathematical medium in which to develop compositionality is category theory. Recently,  Compositional Game Theory~\cite{julesGames} was introduced to bring the advantages of compositionality to game theory. However, within compositional game theory there are a number of different notions of morphisms of open games and, as we shall see, certain notions behave better with regards to approximation than others. Finding the right notion of morphisms of open games to interact well with approximation was another issue we had to solve.\\[1ex]

\noindent Concretely, this paper will
\begin{itemize}
    \item introduce a notion of the $\epsilon$-approximation of a game by defining for any game $G$, another game $\T_{\epsilon} G$ which we think of as containing not just the equilibria of $G$, but also the $\epsilon$-equilibria of $G$ --- this is an entirely new approach to approximation of equilibria and shows novelty of methodology.
    \item show that $\T_{\epsilon}$
has good algebraic structure in that 
it is i) functorial; and ii) possesses graded structure --- this shows our methodology is canonical in being aligned with other work on quantitative computation.\cite{tarmo} 
\item show $T_{\epsilon}$ interacts well with the compositional structure of game theory, in particular with Nash equilibria and sequential composition. There is no particular reason that the approximation of a large game built from game-constructors should retain that structure --- we show this holds for the key cases of parallel (or Nash) product and sequential composition. 
This is crucial evidence for the utility of our approach, namely that preserves the compositional structure of game theory.
\item use our approximate game theory to create what one might call {\em Metric Game Theory} which i) not only builds new games from old; but ii) ensures if one substitutes a subgame with another which has similar behaviour, the overall system remains similar to the original. The crucial construction permitting {\em Metric Game Theory} is a metric on games which we build from our approximate game theory.
\end{itemize} 
\noindent This third item is particularly surprising and thus our strongest contribution --- as mentioned there is no inherent reason why, for example, the approximations of a Nash equilibria should form a Nash equilibria. Finally, we argue in the conclusions that our methodology is also general in that it can apply to other forms of game theory, eg Markov Decision Processes, and beyond game theory into the huge area of machine learning. This latter application area is particularly exciting as formally quantifying distance from optimality is a way of expressing how much has been learned and thus contributing to the XAI agenda.  \\[1ex]

\noindent This paper is structured as follows. In Section 2, we cover preliminaries pertaining to metric spaces, lenses and selection functions. Section 3 develops approximation for selection functions, discusses the algebraic properties of approximation and shows approximation is compositional. Section 4 extends these results to open games~\cite{julesGames}, while Section 5 uses approximation to develop Metric Game Theory. Section 6 concludes the paper with conclusions and directions for future research. 

\section{Preliminaries: Metric Spaces and Lenses}

This paper should be accessible to those with a basic understanding of game theory and category theory. Introductory texts are~\cite{watson2002strategy,eogt,or94}, for game theory and ~\cite{mac71,jaco16:intro} for category theory.
 
\subsection{Metric Spaces}  Let $\R$ be the positive reals extended with an infinity $\infty$
\begin{definition} Define:
  \begin{enumerate}
    \item An extended metric space is a set $X$ with a function $d: X \times X \to \R$ satisfying symmetry, the triangle inequality and $d(x,y) = 0$ iff $x=y$.  
    \item \emph{short} maps between metric spaces $(X,d)$ and $(Y,d')$ are functions $f:X \to Y$ such that $d(fx,fy) \leq d(x,y)$.
  \end{enumerate}
\end{definition}
\noindent We write $\abs{X}$ for the underlying set of a metric space $X$. 

\begin{definition}
  If $(X,d)$ and $(Y,d')$ are extended metric spaces, then $(X,d) \tensor (Y,d')$ is the extended metric space with underlying set $X \times Y$, where
  \[d_{X \tensor Y}((x,y),(x',y')) = \msf{max}\{ d(x,x') , d'(y,y')\} \]
\end{definition}

\noindent Other options are possible, eg \msf{sum}, average etc. The following explains the use of $\infty$ as a possible distance --- it is needed to ensure the metric on function spaces is well defined.
\begin{definition}
  Let $X,Y$ be a metric spaces. 
  The set of functions $[X,Y]$ is a metric space via the \emph{sup-metric}:
  \[d(f,g) = \sup_{x \in X} \; d(fx,gx).\]
\end{definition}

\subsection{Lenses}
Lenses underpin selection functions and open games and so we introduce them here. While lenses can be built over categories with appropriate structure, to keep our presentation simple, we work with lenses over sets. 

\begin{definition}
A lens is given by a pair of sets $(X,V)$. A map of lenses $f:(X,V) \to (Y,Z)$ is given by a pair $(f_0,f_1)$ consisting of functions $f_0:X \to Y$ and 
$f_1 : X \times Z \to V$.     
\end{definition}

\noindent The following is well known~\cite{julesGames,julesPhD}:

\begin{lemma}
Lenses from a symmetric monoidal category $\Lens$ with tensor given by $(X,V) \otimes (X',V') = (X \times X', V \times V')$ and with unit given by $I = (1,1)$.
\end{lemma}

\noindent Note i) an element of $X$ is a lens-map $I \to (X,V)$; and ii) a function $X \to V$ is equivalently a lens map $(X,V) \to I$. These two observations are useful, eg we can often represent complex functions $X \to V$ more simply as composites of lens maps.

\subsection{Selection Functions}

\noindent A basic form of game theory is given by multi-valued selection functions~\cite{escardo-oliva} which, given sets $X,Y$, are defined by $\msf{Sel}(X,Y) = (X \to Y) \to PX$ where $PX$ is the set of powersets of $X$. We recall that for simplicity we just call them selection functions henceforth. 
Intuitively, a selection function maps a utility function (which assigns to each potential action $x \in X$ to a resulting utility $y \in Y$) to a set of actions deemed to be optimal for that specific utility function.
Give two selection functions $s,t:(X \to V) \to PX$, we write
$s \subseteq t \;\; \msf{iff} \;\; (\forall \; f \!: \!X \to V) \;\; sf \subseteq tf $. The canonical selection function is $\AM$ defined as follows:

\begin{example}
    Define the argmax selection function $\AM:(X \to Y) \to PX$ by $$\AM f \; = \; \{ \;\; x \in X \;\; | \;\; (\forall x' \in X) \;\; fx \geq f x'\}$$
\end{example}

\noindent Note that the selection function model of game theory deals with not merely one utility function but, rather, all utility functions. This is also the case of open games and is a hallmark --- indeed key distinguishing feature --- of compositional models of game theory. When it comes to maps of selection functions, we choose the following for reasons that will become clearer later.

\begin{definition}
    Let $s : \Sel(X,V)$ and $t:\Sel(X',V')$. A map from $s$ to $t$ 
    consists of a map of lenses $\alpha: (X,V) \to (X',V')$ such that 
    \[
(\forall k:X' \to V')(\forall x :X) \;\; \alpha x \in t k \Rightarrow x \in s (k\alpha)
    \]    
\end{definition}
\noindent  Note in the above we write the map $X \to V$ as the lens-composite of lens maps $k\alpha$ as this is much cleaner than a direct definition. This is an one example of the power of lens composition to simplify definitions as we remarked earlier. We continue to use lens composition via juxtaposition through the rest of the paper without remarking on it. The above definition also possesses a contravariance akin to how maps of containers (aka dependent lenses) are contravariant in their position maps~\cite{containers}.
An alternative covariant definition for a map of selection functions would have been to require 
 \[
(\forall k:X' \to V') \;\; x \in s (k\alpha)  \Rightarrow \alpha x \in t k
    \] 
but, as we shall see later, this does not give a functorial theory of approximation. Henceforth we write $\Sel$ for the category of all selection functions and the maps between them. 
Selection functions are a compositional model of game theory as they support a Nash product. This arises by extending the monoidal product on lenses to one on selection functions:
\begin{definition}
    The category $\Sel$ of selection functions and maps of selection functions is monoidal with monoidal product, also called the Nash product, $$\otimes: \msf{Sel}(X,Y)\times \msf{Sel}(X',Y') \to \msf{Sel}(X\times X',Y\times Y')$$ defined by 
    \begin{eqnarray*}
        (x,x') \in (s \otimes t) f & \msf{iff} & x \in s (\pi_0 f(-,x')) \\
                                  &           & \wedge \\
                                &             & x' \in t (\pi_1 f(x,-))
    \end{eqnarray*}
\end{definition}

\noindent Note as an example, we have the beautiful equation showing Nash equilibria are not a primitive but rather a derived concept.
\[
\msf{Nash} = \msf{argmax} \otimes \msf{argmax} 
\]
See~\cite{julesPhD} for more details. This equation also highlights how our product of selection functions can be seen concretely as the two games played in parallel with the famous Nash equilibrium of Prisoners Dilemma as the specific instance of two argmaxes run in parallel. 
In developing approximate game theory, we ensure it is compositional in that approximation interacts well with compositional operators such as the Nash product. 

\section{Approximation and Selection Functions}
Given selection functions form a compositional model of game theory we can ask two natural questions: i) can we account for approximate equilibria in selection functions; and ii) is approximation compositional in that --- for example --- if we approximate a Nash product of two games, do we get the Nash product of two approximations? We answer both questions positively --- the second to our surprise as there is no particular reason to believe that approximation of a structured game preserves that structure. The fact it does suggests the compositional approach to game theory and the approach to approximation advocated here are well designed concepts capturing fundamental logical structure which other approaches have missed. We begin by defining approximation in selection functions via a function
\[
\msf{T} : (\epsilon : \R) \to \Sel(X,V) \to \Sel(X,V)
\]
by
\[
x  \in (\msf{T}_{\epsilon} G)  k \;\;\; \msf{iff} \;\; \;(\exists k' \sim_{\epsilon} k) \;\; x \in G k'
\]
where we write $k' \sim_{\epsilon} k$ for $d(k',k) \leq \epsilon$ using the sup-metric on functions defined above. What is nice about this definition is that not only is it extremely simple and thus elegant, but it also formalises the intuition that an $\epsilon$-optimal choice for a utility function $k$ is an actual optimal choice, but for a utility function $k'$ differing no more than $\epsilon$ from $k$. Getting this definition right, and getting the right notion of morphism between selection functions, were the fundamental keys required to achieving the smooth mathematical theory we have been able to develop in this paper. The first thing to do is to relate our notion of approximation to the traditional notion of approximation that one finds in the literature.

\begin{lemma}
   Let $\epsilon \geq 0$. Then $\T_{\epsilon}(\AM) \subseteq 2\epsilon\!\!-\!\!\AM$
\end{lemma}

\begin{proof}
    Let $x \in \T_{\epsilon}(\AM)k$ for some utility function $k:X \to V$. Then there is a $k' \sim_{\epsilon} k$ such that $x \in \AM(k')$, ie for all $x' \in X$, we have $k'x \geq k'x'$. But since $k' \sim_{\epsilon} k$, we thus have 
    $kx \geq kx' - 2\epsilon$. Thus $x \in 2\epsilon\!\!-\!\!\AM \; k$
\end{proof}

\noindent Assuming $X$ has a decidable equality, which includes the key case where $X$ is finite, we can go further. Note the use of decidable equality is to ensure $k'$ below is well defined.

\begin{lemma}
\label{lem:approx-em}
Let $\epsilon \geq 0$. If $X$  has a decidable equality, then
\[
\epsilon\!\!-\!\!\AM \subseteq \T_{\epsilon}(\AM)
\]
\end{lemma}
\begin{proof}
Let $x \in \epsilon\!\!-\!\!\AM(k)$ for a utility function $k:X \to V$. Then forall $y \in X$, $k x \geq ky - \epsilon$. Define $k':X \to V$ by $k' x = k x + \epsilon$ and $k' y = k y$ for all $y \neq x$. Clearly $k \sim_{\epsilon} k'$. Further, for any $y\neq x$, we have $k' x \geq k'y$. Thus  $x \in \AM(k')$ and so $x \in 
\T_{\epsilon}(\AM)(k)$
\end{proof}

\noindent Moving to a categorical analysis, the first thing we would like to establish is the functoriality of $\T_{\epsilon}$. This is non-trivial as the needs of approximation put constraints on how lens maps back-propagate values. In essence, we must restrict to {\em short} lens maps which we define now 

\begin{definition}
    A short lens map is a lens map $\alpha :(X,V) \to (X',V')$ where, for every $x \in X$, $\alpha(x,-):V' \to V$ is a short map, ie non-expanding.  
\end{definition}

\noindent Short maps are a natural condition --- the condition reflects desire that coutility (or backpropagation) doesn't inflate differences in the values to be backpropagated. Were we to allow such inflation, system behaviour could well become chaotic as small amounts of sub-optimality would be magnified into large system recongifuratation. Note the identity is a short map and the composite of short maps is short and hence the category of lenses restricts to a category of lenses and short maps denoted $\Lens_s$. Similarly, we define the category of selection functions and maps whose underlying lens map is short by $\Sel_s$. And, finally, the tensor product on lenses and selection functions restricts to a monoidal structure on $\Sel_s$ 

\begin{lemma}
    Let $\alpha:(X,V) \to (X',V')$ be a short lens map. If $k, k':X' \to V'$ is such that $k \sim_{\epsilon} k'$, then $k\alpha \sim_{\epsilon} k'\alpha$  
\end{lemma}

\begin{proof}
    Direct from the definition of what a short lens map is.
\end{proof}

\noindent We can now prove functoriality:

\begin{lemma}
    Let $\alpha: s \to t$ be a map of selection functions over the short lens map $\alpha:(X,V) \to (X',V')$ of lenses. Then $\alpha$ also defines a map of selection functions $ \T_{\epsilon}s \to \T_{\epsilon}t$ 
\end{lemma}

\begin{proof}
    Consider $x \in X$ and $k:X' \to V'$. We need to show if $\alpha x \in (\T_{\epsilon}t) k$, then 
    $x \in (\T_{\epsilon}s) (k\alpha)$. The assumption implies there is a $k_1:X' \to V'$ such that $k \sim_{\epsilon} k_1$ and $\alpha x \in t k_1$. As $\alpha:s \to t$, this means $x \in s (k_1\alpha)$.
    By the previous lemma, $k\alpha \sim_{\epsilon}k_1\alpha$ and so $x \in \T_{\epsilon}s (k\alpha)$ as required. This shows that any map $\alpha:s \to t$ induces a map $\alpha:T_{\epsilon} s \to T_{\epsilon} t$. Preservation of composition and identities is now trivial.
\end{proof}

\noindent Next we consider the role/structure of the parameter $\epsilon$ in approximation. The following graded structure is what one expects in many quantitative situations.

\begin{lemma}
\label{lem:graded}
The following are easily provable
\begin{itemize}
    \item $\T_0 G = G$
    \item If $\epsilon \leq \epsilon'$, then $\T_{\epsilon} G \subseteq  \T_{\epsilon'} G$
    \item $\T_{\epsilon} (\T_{\delta }G) \subseteq  \T_{\epsilon + \delta} G$
\end{itemize}
\end{lemma}

\begin{proof}
Straightforward
\end{proof}


\subsection{Compositional Approximation}
In this subsection we discuss the relationship between the compositional structure of selection functions and approximation of selection functions. To do this we first prove a small lemma

\begin{lemma}
Let $k, k': X\times X' \to V \times V'$ be functions such that $k \sim_{\epsilon} k'$. Then for every $x$ and
$x'$, 
\[
\pi_0k(-,x') \;\; \sim_{\epsilon} \;\; \pi_0k'(-,x') \;\; \mbox{ and } \;\;
\pi_1k(x,-) \;\; \sim_{\epsilon} \;\; \pi_1k'(x,-)
\]
\end{lemma}

\begin{proof}
Straightforward calculation
\end{proof}

\noindent Now, consider $G \otimes H$, the Nash product of selection functions $G$ and $H$. What do its approximate equilibria look like? Is there any structure here or does it just consist of a fuzzy ball of extra equilibria enveloping $G$ and $H$? The following shows that structure is preserved, ie that an approximate equilibrium of a Nash product is infact an equilibrium of a Nash product.

\begin{lemma}
    Let $G:\Sel(X,V)$ and $H:\Sel(X',Y')$ be selection functions. Then
    \[
    \T_{\epsilon} (G \otimes H) \subseteq  \T_{\epsilon} G \otimes T_{\epsilon} H
    \]
\end{lemma}

\begin{proof}
    Let $(x,x') \in X \times X'$ and $k:X\times X' \to V\times V'$. Further assume $(x,x') \in \T_{\epsilon}(G \otimes H)\; k$. Then there is a $k':X\times X' \to V \times V'$ such that $k \sim_{\epsilon} k'$ and $(x,x') \in (G \otimes H)\; k'$. By definition of $G \otimes H$, this means 
\[
x \in G \; (\pi_0 k'(-,x')) \;\; \mbox{ and } \;\; 
x' \in H \; (\pi_1 k'(x,-))
\]
By the previous lemma
\[
x \in \T_{\epsilon}(G) \; (\pi_0 k(-,x')) \;\; \mbox{ and } \;\; 
x' \in \T_{\epsilon}(H) \; (\pi_1 k(x,-))
\]
and hence $(x,x') \in (\T_{\epsilon}G \; \otimes \;\T_{\epsilon}H)\; k$. Thus
$T_{\epsilon} (G \otimes H) \subseteq  \T_{\epsilon} G \otimes T_{\epsilon} H$ as required.
    \end{proof}

\noindent The above lemma states that the approximation of a Nash product is itself a Nash product.  Noting the contravariance in our notion of morphism of selection functions, we can summarise the above as:

\begin{lemma}
\label{lem:lax-monoidal}
    The functor $\T_{\epsilon}:\Sel_s \to \Sel_s$ is lax monoidal functor
\end{lemma}

\noindent However, it would be nice if the reverse were true as then we could {\em find} and approximate equilibrium for a Nash equilibria by merely approximating the components. Not suprisingly, the ability to take this extra step depends on the existence of sufficient utility functions as we saw in lemma~\ref{lem:approx-em} which in turn rests on
the decidability of equality on actions.  The key lemma is the following which uses the decidable equality to ensure the function $k'$ below is well defined.

\begin{lemma} 
\label{lem:dec-monoidal}
Let $X,X'$ have decidable equality. Further, let $k:X \times X' \to V \times V'$, $x \in X$ and $x' \in X'$. Even further, let $g \sim_{\epsilon} \pi_0 k(-,x')$ and $h\sim_{\epsilon} \pi_1 k(x,-)$. Then
    there is a $k':X\times X' \to V \times V'$ such that
    \begin{itemize}
        \item $k' \sim_{\epsilon} k$
        \item $g = \pi_0 k'(-,x')$; and 
        \item $h = \pi_1 k'(x,-)$.
    \end{itemize}
\end{lemma}

\begin{proof}
    Define $k':X \times X' \to V \times V'$ as follows
    \begin{eqnarray*}
        k'(z,z') & = & k(z,z')  \hspace{0.8in} \mbox{if} \;\; z \neq x, z' \neq x' \\
        k'(x,z') &= & (\pi_0k(x,z'), hz')  \hspace{0.3in} \mbox{if} \;\;  z' \neq x' \\
        k'(z,x') & = & (gz, \pi_1k(z,x'))  \hspace{0.3in} \mbox{if} \;\; z \neq x \\
        k'(x,x') & = & (gx,hx') \hspace{0.7in}\mbox{otw} 
     \end{eqnarray*}
     The required properties of $k'$ are now easy to verify.
\end{proof}

\noindent Now we can prove the reverse containment

\begin{lemma}
    Let $X,X'$ have decidable equality. Let $G:\Sel(X,V)$ and $H:\Sel(X',Y')$ be selection functions. Then
    \[
    \T_{\epsilon} G \otimes T_{\epsilon} H \subseteq \T_{\epsilon} (G \otimes H) 
    \]
\end{lemma}

\begin{proof}
    Let $(x,x') \in X \times X'$ and $k:X\times X' \to V \times V'$ be such that $(x,x') \in (\T_{\epsilon} G \otimes T_{\epsilon} H) k$. Then 
    \[
x \in (\T_{\epsilon} G) (\pi_0k(-,x')) \;\; \mbox{ and } \;\; 
x' \in (\T_{\epsilon} H) (\pi_1k(x,-))
    \]
    Thus there is a function $g:X \to V$ and one 
    $h:X' \to V'$ such that 
    \[
x \in G g \;\; \mbox{ and } \;\; x' \in H h
\;\; \mbox{ and } \;\;
g \sim_{\epsilon} \pi_0k(-,x')
\;\; \mbox{ and } \;\; 
h \sim_{\epsilon} \pi_1k(x,-)
    \]
    By lemma~\ref{lem:dec-monoidal}, there is a map $k':X\times X' \to V \times V'$ such that 
    \[
 x \in G (\pi_0k'(-,x')) \;\; \mbox{ and } \;\; 
x' \in H (\pi_1k'(x,-))
\;\; \mbox{and} \;\;
k' \sim_{\epsilon} k
    \]
Thus $(x,x') \in (G \otimes H) k'$ and hence 
$(x,x') \in \T_{\epsilon}(G \otimes H) k$ 
    \end{proof}

\noindent This is a very strong --- and surprising result --- that for action sets with decidable equality, approximation is not merely lax monoidal, but actually monoidal. This we did not expect and, indeed, when it comes to sequential composition of open games, we can only show one containment. Sequential composition of games leads us from selection functions to Open Games.

\section{Approximation and Open Games}
Selection functions form a simple model of compositional game theory via the Nash product/parallel composition of selection functions. However, what about the sequential composition of selection functions --- 
after all parallel and sequential composition are intimately bound in monoidal category theory, string diagrams etc. Sequential composition of selection functions has been considered~\cite{escardo-oliva} but the utility function of the first game is modelled exogenously ... what happens in practice is that the utility function of the first game is actually computed endogenously. Open Games~\cite{julesGames} address this problem perfectly. For the price of a little more algebraic structure, they form a compositional model of game theory which includes a number of operators for compositionally building large complex games from smaller and simpler games. Chief amongst these operators are sequential composition but also choice operators, and operators for iterated game theory~\cite{GhaniKLF18}. This section generalises approximation to open games.

\subsection{Open Games}

The key concept which introduced open games is the following:
\begin{definition}[Open Game]
Let $X, Y, R$ and $S$ be sets. An open game
$G = (\Sigma_G,P_G,C_G,E_G) : (X,S) \to (Y,R)$ consists of:
\begin{itemize}
    \item a set $\Sigma_G$ of strategy profiles,
    \item a play function $P_G : \Sigma_G\to (X \to Y )$,
    \item a coutility function $C_G : \Sigma_G \to (X\times R \to S)$, and
  \item an equilibrium function $E_G :X \times (Y \to R) \to \P(\Sigma_G)$. 
\end{itemize}

\end{definition}
\noindent Intuitively, the set $X$ contains the states of the game, $Y$ the moves, $R$ the utilities and $S$ the coutilities. The set $\Sigma_G$ contains the strategies we are trying to pick an optimal one from. The play function $P_G$ selects a move given a strategy and a state, while the coutility function $C_G$ computes the coutility extruded from the game, given a strategy, state and utility. Finally, if $\sigma \in E_G \; x \; k$, then $\sigma$ is an optimal strategy in state $x$ and with utility given by $k : Y \to R$. Just as lenses systematise the presentation of selection functions, they also systematise the presentation of open games as can be found in~\cite{julesPhD,julesGames}. 

\begin{lemma}
    An open game $G : (X, S) \to (Y, R)$ is given by a set $\Sigma$ of strategies and, for each $\sigma \in \Sigma$: i) a lens $G_\sigma:(X,S) \to (Y,R)$; and ii)
    a predicate
    \[
E_{\sigma} \subseteq X \times (Y \to R)
    \]
\end{lemma}

\noindent We now restrict all lenses occurring in open games to be short. This is important not only in functoriality as before, but also for the compositional treatment of approximation for sequential games. Notice the similarity of the lens predicate $E_{\sigma}$
and a selection function. But notice also differences --- a selection function is defined over a pair of sets, ie an object of $\Lens$, while an open game is defined over a family of lens maps. There are other notions of open games, in particular Gavranovic Games ~\cite{cybernetics} which deepen the relationship with lenses by defining open games to be parameterised maps in the category of lenses. We expect our techniques apply to Gavranovic games in the obvious way.

\subsection{Approximation and Open Games}

The methodology developed for selection functions clearly generalises to open games as follows:

\begin{definition}
    Let $G:(X,S) \to (Y,R)$ be an open game with strategy set $\Sigma$ and equilibrium predicates $E_{\sigma}$ for $\sigma \in \Sigma$. Further let 
    $\epsilon \in \R$. Then define $\T_\epsilon G:(X,S) \to (Y,R)$ to be the open game with i) the same strategy set $\Sigma$; ii) lens structure 
    $(\T_\epsilon G)_{\sigma} = G_{\sigma}$; and iii) equilibrium given by 
    \[
x \in E_{\T_{\epsilon}G, \sigma} \; k 
\;\;\mbox{ iff} \;\; \;(\exists k' \sim_{\epsilon} k) \;\; x \in E_{G, \sigma} k'
\]    
\end{definition}

\noindent The similarity of approximation for open games with that for selection functions is obvious and highlights a strength of the methodology developed here --- it is likely to generalise to many other situations as discussed in the conclusions and future work. This similarity also means that lemma~\ref{lem:graded} carries over to the setting of open games. Next we consider functoriality of approximation for open games


\begin{definition}
    The category of open games $\msf{Op}$ has open games as objects, and a morphism from the open game $(\Sigma,G,E):(X,S) \to (Y,R)$ to the open game $(\Sigma', G', E'):(X',S') \to (Y',R')$ consists of a function $f: \Sigma \to \Sigma'$, short lens maps 
    $\alpha : (X,S) \to (X',S')$ and $\beta:(Y,R) \to (Y',R')$ such that 
        \[
        (\forall \sigma \in \Sigma)(\forall x\in X)(\forall k: Y' \to R')
        \;\; \alpha x \in E'_{f\sigma} k \;\; \mbox{implies} \;\;
        x \in E_{\sigma} (k . \beta)
        \]
\end{definition}
    
\noindent Notice how the contravariance in maps of selection functions is replicated here in maps of open games. This will be another hallmark of a general theory of approximation for systems consisting of utility indexed predicates. And again notice the use of the shortness of the lens maps $\alpha$ and $\beta$. The proof of the following lemma is the natural generalisation of the analogous one for selection functions.

\begin{lemma}
For $\epsilon \in \R$, the assignment sending an open game $G$ to an open game $\T_{\epsilon}G$ extends to a functor $\T_{\epsilon}:\msf{Op} \to \msf{Op}$ 
\end{lemma}

\noindent The final part of the mathematical development of approximation for selection functions considered the monoidal product formalising Nash equilibria which showed that $\T_{\epsilon}(G \otimes H) \subseteq \T_{\epsilon}G \otimes \T_{\epsilon}H$, and that $\subseteq$ can be replaced by $=$ 
when moves are decidable. Again, this lemma generalises to open games --- see~\cite{julesGames} for exact the definition of
the Nash product for open games.

\begin{lemma}
    The functor $\T_{\epsilon}:\msf{Op} \to \msf{Op}$ is lax monoidal with respect to the Nash tensor on open games, and monoidal where the type of moves (inputs to utility) is decidable.
\end{lemma}

\begin{proof}
    The proof follows the same structure as Lemma~\ref{lem:lax-monoidal}.
\end{proof}

\noindent One of the great advantages of open games is they compositionally model the sequential composition of games. Indeed, if we were to form a category whose objects are pairs of sets and whose morphisms are open games, the sequential composition of open games would be the composition in that category.

\begin{definition}
    Let $(\Sigma, G, E) :(X,S) \to (Y,R)$ be an open game and $(\Omega, H, B) : (Y,R) \to (Z,T)$ be another open game. Their sequential composition is the open game $H\circ G : (X,S) \to (Z,T)$ defined by
    \begin{itemize}
        \item Strategies of $H\circ G$ are $\Sigma \times \Omega$
        \item The lens $(H\circ G)(\sigma,\tau)$ is defined to be the lens composition  $H\tau . G\sigma$
        \item The equilibirum predicate is defined by
        \[
x \in E_{HG, (\sigma,\tau)} (k:Z \to T) \;\;\; \mbox{iff} \;\;\; x \in E_{G, \sigma} (k . H\tau) \;\; \wedge \;\;
G_{\sigma} x \in E_{H,\tau} \; k
        \]
    \end{itemize}
\end{definition}

\noindent Just like with the Nash product, we can ask ourselves
the if the approximate equilibria of a sequential game has any structure, eg is it composed of equilibria for the sequential composition of two approximated games. Once more, the answer is yes!

\begin{lemma}
    Let $(\Sigma, G, E) :(X,S) \to (Y,R)$ and $(\Omega, H, B) : (Y,R) \to (Z,T)$ be open games. Then
    \[
x \; \in \; E_{\T_{\epsilon} (H \circ G), (\sigma, \tau)} \; (k : Z \to T) \;\; 
\mbox{implies} \;\;  
x \; \in \; E_{\T_{\epsilon}H \circ T_{\epsilon} G, (\sigma,\tau)} \; (k : Z \to T) 
    \]
    for any $x\in X$, $k:Z \to T$, $\sigma \in \Sigma$ and $\tau \in \Omega$
\end{lemma}

\begin{proof}
 Assume $x\in X$, $k:Z \to T$, $\sigma \in \Sigma$ and $\tau \in \Omega$ and $x \; \in \; E_{\T_{\epsilon} (H \circ G), (\sigma, \tau)} \; (k : Z \to T)$. By definition of $\T_{\epsilon}$, we have a $k':Z \to T$ such that $k \sim_{\epsilon} k'$ and  $x \; \in \; E_{H \circ G, (\sigma, \tau)} \; (k' : Z \to T)$. By definition, this means 
$x \in E_{G, \sigma} (k' . H\tau)$ and $G_{\sigma} x \in E_{H,\tau} \; k'$. Since $H\tau$ is short, we have  
$k.H\tau \sim_{\epsilon} k'.H\tau$ and hence 
\[
x \in E_{\T_{\epsilon}G, \sigma} (k . H\tau) \;\; \mbox{ and } \;\;  G_{\sigma} x \in E_{\T_{\epsilon}H,\tau} \; k
\]
Thus $x \; \in \; E_{\T_{\epsilon}H \circ T_{\epsilon} G, (\sigma,\tau)} \; (k : Z \to T)$ 
\end{proof}

\noindent Unlike with the Nash product, it is unlikely we can prove the reverse direction as, given $x \in E_{\T_{\epsilon}G, \sigma} (k . H\tau)$, all we could conclude is there is a $k' \sim_{\epsilon} k . H\tau$ such that $x \in E_{G,\sigma} k'$ but vital and needed structure is lost.

\section{Metric Game Theory}
Apart from solving the concrete problem of when is a move  almost optimal in a selection function or open game, our Approximate Game Theory permits further theoretical developments. The one we sketch here we call {\em Metric Game Theory} which seeks to ask the question of how far apart are two games. This is important, eg in a compositional model we might want to know if we replace sub-games/components of one game with other sub-game/components, can we quantitatively measure the effect. In more detail, we might ask questions such as if $G \sim_{\epsilon} G'$, what can we say about the distance between $G \otimes H$ and $G' \otimes H$.
In this section we develop the basics of Metric Game Theory. We focus on selection functions as the theory seems to generalise to open games. We also only measure distances between selection functions over the same lens --- we expect we can use the fibrational structure of selection functions over lenses to enable the computation of  distances between selection functions over different lenses. But we leave a proper treatment of {\em Metric Game Theory} to a full paper as the subject clearly warrants. 

\begin{definition}
    Let $s,s' \in \Sel(X,V)$ be selection functions. We define  
     \[
    d(s,s') = (\msf{inf} \; \epsilon : \R) \;\; (s\subseteq T_{\epsilon} s' \;\; \wedge \;\; s'\subseteq T_{\epsilon} s)
    \]
\end{definition}
\noindent In the above, we are essentially using the Hausdorff distance on subsets. As a result, the above construction gives a pseudometric, ie there is no requirement that $d(x,y) = 0$ implies $x=y$. We list two obvious properties of the pseudometric on $\Sel(X,V)$ showing how distances between games are preserved by compositional structure. 

\begin{lemma}
    Assuming a decidable equality on $X,X'$, the following are true
    \begin{itemize}
        \item Let $s$ be a selection function. Then $d(s,\T_{\epsilon}s) \leq \epsilon$
        \item Let $s,s' \in Sel(X,V)$ and $t,t' \in \Sel(X',V')$. If $d(s,s') \leq \epsilon$ and $d(t,t') \leq \epsilon$, then $d(s \otimes t,s'\otimes t') \leq \epsilon$
        \end{itemize}
\end{lemma}

\begin{proof}
    The first property follows from the graded structure of $\T_{\epsilon}$. For the second, note $s \subseteq \T_{\epsilon}s'$ and $t \subseteq \T_{\epsilon}t'$. Thus $s\otimes t \subseteq \T_{\epsilon}s' \otimes \T_{\epsilon}t' = \T_{\epsilon}(s' \otimes t')$
\end{proof}

\noindent Finally, note this states if we replace a subcomponent of a Nash product by one $\epsilon$-different, the overall Nash product is no further than $\epsilon$-apart. This is an ideal property to have in a quantitative compositional theory of games. The analogous property for the Nash product of open games also holds. However, note that for sequential composition of open games we only have $\T_{\epsilon}(H\circ G)\subseteq \T_{\epsilon}H \circ T_{\epsilon}G$. Thus, replacing subgames with other subgames can increase distance between the overall games but (not shown here) in a bounded way --- enough to give a theory of {\em uncertainty budgets} for approximation of open games. 

\section{Conclusions and Future Work}

In summary, we have tackled the important problem of approximation for game theory. This is important as exact equilibria are often either impossible to determine due to inherent uncertainty in underlying data, or undesirable because the cost of exact computation is too high in comparison to approximate computation. In developing our approach to approximate game theory, it was vital to ensure our theory of approximation was compositional in that it respected the compositionality of game theory. Our fundamental idea was that a strategy was approximately optimal for a given utility function if it is optimal for an approximation of that utility function. This idea was developed for two models of game theory - firstly selection functions and secondly open games. In both, approximation was shown to preserve compositional structure. Finally, we showed how approximation and compositionality can be used to develop {\em Metric Game Theory} which allows one to replace approximately equal subcomponents of a Nash product and produce guarantees on how close the overall resulting systems are. \\[1ex]

\noindent By opening up a new field, this paper has much that can be done to expand it. Firstly, there are a number of other operators in open games which the methods developed here could be applied to. These include choice, contextualisation and infinite iteration~\cite{GhaniKLF18}. Secondly, these methods could also be extended to more advanced forms of open games, eg i) probabilistic open games using distance metrics such as Kantorovich on probability distributions; or ii) games based upon optics~\cite{riley2018categories}. Thirdly, the generality of our ideas means applications to other forms of games are natural For example, in Markov Decision Processes~\cite{hansen}, one can consider coalgebras
\[
\langle u,t \rangle : X \to R \times (A \to DX)
\]
where $X$ is a state space, $R$ is a type of utility, $A$ is a set of actions, $D$ is the probability monad, $u:X \to R$ is a utility function and $t:X \times A \to DX$ is a probabilistic transition function. Within MDPs one defines optimal policies which we can denote $\msf{Opt}$.
Our methodology suggests defining an approximate version
\[
\pi \in \msf{Opt} \; T_{\epsilon}\langle u,t\rangle \;\; \msf{iff} \;\;
\pi \in \msf{Opt} \; \langle u',t\rangle 
\]
where $u' \sim_{\epsilon} u$. One could even be more general and consider tweaks to the structure of $t$! Most excitingly, machine learning has features in common with game theory in that the former finds a parameter to minimise a loss function while the later finds a strategy to  maximise a utility function~\cite{cybernetics,catgrad}. Thus it is possible that the ideas developed here could align with 
learning based upon back propagation~\cite{baf} and, in particular, categorical deep learning~\cite{gavranović2024categorical} and its special case of geometric deep learning. This would contribution to the huge problem of XAI. It would be good to see what an implementation in code of our ideas looks like, eg in the implementation of open games~\cite{julesHaskell}. The breadth of application areas also suggests more theoretical work unifying these disparate examples, eg via a theory of approximation for utility indexed predicates, or maybe predicates fibred over metric spaces.

\bibliographystyle{eptcs}
\bibliography{generic}
\end{document}